\documentclass[11pt]{article}

\usepackage{booktabs} 
\usepackage[T1]{fontenc}
\usepackage[latin9]{inputenc}
\usepackage[letterpaper]{geometry}
\usepackage{amsmath}
\usepackage{amsthm}
\usepackage{amssymb}

\makeatletter

\newtheorem{theorem}{Theorem}[section]
\newtheorem{corollary}[theorem]{Corollary}
\newtheorem{lemma}[theorem]{Lemma}
\newtheorem{proposition}[theorem]{Proposition}
\newtheorem{claim}[theorem]{Claim}
\newtheorem{definition}[theorem]{Definition}

\usepackage{epigraph}

\usepackage{etoolbox}

\patchcmd{\epigraph}{\@epitext{#1}}{\itshape\@epitext{#1}}{}{}

\makeatother

\newcommand{\hp}{high priced}

\begin{document}

\title{Revenue Loss in Shrinking Markets}

\author{Shahar Dobzinski \and Nitzan Uziely}
%
%



\maketitle
\begin{abstract}
We analyze the revenue loss due to market shrinkage. Specifically,
consider a simple market with one item for sale and $n$ bidders whose
values are drawn from some joint distribution. Suppose that the
market shrinks as a single bidder retires from the
market. Suppose furthermore that the value of this retiring bidder
is fixed and always strictly smaller than the values of the other
bidders. We show that even this slight decrease in competition might
cause a significant fall of a multiplicative factor of $\frac{1}{e+1}\approx0.268$
in the revenue that can be obtained by a dominant strategy ex-post
individually rational mechanism.

In particular, our results imply a solution to an open question that
was posed by Dobzinski, Fu, and Kleinberg {[}STOC'11{]}. 
\end{abstract}

\epigraph{''Take heed of the children of the poor, for from them will Torah come forth.''\\\hspace{2.5in}--- Babylonian Talmud, Nedarim 81a.}

\section{Introduction}

How much revenue might a firm lose due to market shrinkage? We study
this question in a simple market with $n$ bidders and one item. The
private values of the bidders $(v_{1},\ldots,v_{n})$ are drawn from
some known distribution. Now suppose that one bidder retires from
the market. Our goal is to compare the maximum revenue that can be
extracted in the original market by a dominant-strategy mechanism
that is also ex-post individually rational to the maximum revenue that
can be obtained by a similar mechanism in the smaller market.

Obviously, if the value of the retiring bidder is always much larger
than the values of the rest of the bidders then almost all revenue
will be lost. Thus, we consider an extreme situation where the retiring
bidder is the weakest competitor in the market: in every realization
$(v_{1},\ldots,v_{n})$ the value of the retiring
bidder $v_{i}$ is always strictly smaller than the value of any other
bidder $v_{j}$. Furthermore, we will assume that $v_{i}$ is identical
in all realizations, so the value of the retiring bidder $v_{i}$
conveys no information at all about the values of other bidders.

One might speculate that as the number of bidders increases the relative
contribution to the revenue of payments by this retiring bidder
diminishes. However, we show that \textendash{} perhaps counter-intuitively
\textendash{} even this slight decrease in competition, i.e., the
same large market but with the absence of its least valuable consumer,
might cause the revenue to fall by a constant multiplicative factor,
independently of the size of the market:

\vspace{0.1in}\noindent  \textbf{Theorem (informal): } For any $n$, there exists a joint distribution ${\mathcal{H}_{n}}$ over the values of $n$ bidders with the following
properties: 
\begin{itemize}
\item Bidder $n$ is a ``weak'' bidder (as discussed above).
\item The maximum expected revenue that can be extracted by a dominant strategy
ex-post individually rational mechanism in a market with $n$ bidders
whose values are distributed according to $\mathcal{H}_{n}$ is at
least $1$. 
\item Let $\mathcal{H}_{n-1}$ be the joint distribution over $n-1$ values that
is obtained from $\mathcal{H}_{n}$ by removing the value of bidder
$n$. The maximum expected revenue that can be extracted by a dominant
strategy ex-post individually rational mechanism in a market with
$n-1$ bidders whose values are distributed according to $\mathcal{H}_{n-1}$
is at most $\frac{e}{e+1}\approx0.731$. 
\end{itemize}
A dual interpretation of our result is that in some markets firms should consider investing effort in market expansion, as even recruiting a single low value consumer might lead to a revenue surge. Obviously there are markets in which recruiting low value consumers does not lead to a significant increase in the revenue -- understanding whether there are practical distributions in which recruiting low value consumers leads to a significant increase in the revenue is an interesting open question.

\noindent As we will discuss later, the $\frac{e}{e+1}$ ratio is 
essentially tight, by a result of \cite{chen2011approximation}. Interestingly,
we cannot hope to obtain a similar result with independent valuations:
when the values are distributed independently and identically, it
is easy to see (directly or by applying market-expansion theorems
like Bulow-Klemperer \cite{bulow1994auctions}) that a removal of
any single bidder decreases the revenue by a factor of at most $\frac{1}{n}$,
where $n$ is the number of bidders in the market. A more careful
argument gives that this factor continues to hold when removing the weakest bidder from a market in which the values are distributed independently but not necessarily identically \cite{ronen2001approximating}.

We note that our theorem holds regardless of whether we compare the
best randomized truthful in expectation mechanisms in both markets
or the best deterministic mechanisms.

\subsection*{Connection to Previous Work}

Our result is directly connected to the literature on approximating
revenue maximizing auctions when the values of the bidders are correlated.
This line of research was initiated by Ronen \cite{ronen2001approximating}.
In particular, Ronen introduces the Lookahead auction: this is the
dominant strategy, ex-post individually rational revenue maximizing
auction among all auctions that are only allowed to sell to the bidder
with the highest value. Ronen shows that the Lookahead auction extracts
in expectation at least half of the expected revenue of the \emph{unconstrained}
dominant-strategy individually rational revenue maximizing mechanism\footnote{Cremer and Mclean \cite{cremer1988full} show that under certain assumptions
on the distribution there is a dominant-strategy mechanism that extracts
\emph{all} the surplus of the bidders. However, their mechanism is
only ex-interim individually rational, whereas our interest here is
in ex-post individually rational mechanisms.}.

The $k$-lookahead auction is a natural generalization: this is the
dominant strategy, ex-post individually rational revenue maximizing
auction among all auctions that are only allowed to sell the item to one of the $k$ bidders
with the highest values. Dobzinski, Fu, and Kleinberg \cite{dobzinski2011optimal}
show that the $k$-lookahead extracts a fraction of at least $\frac{2k-1}{3k-1}$
of the revenue of the unconstrained revenue maximizing mechanism.
The analysis was improved by \cite{chen2011approximation} where it
was shown that the fraction is at least $\frac{e^{1-\frac{1}{k}}}{e^{1-\frac{1}{k}}+1}$
and that this is tight for $k=2$.

However, a question that was left open in \cite{dobzinski2011optimal,chen2011approximation}
is to determine whether the expected revenue of the $k$-lookahead
auction approaches $1$ as $k<n$ increases\footnote{Previously, for $k>2$ the best result was that for every $k$ there
is a distribution for which the $k$-lookahead auction does not extract
more than $\frac{k}{k+2}$ of the revenue \cite{chen2011approximation},
improving over the $\frac{k}{k+1}$ factor obtained by \cite{dobzinski2011optimal}.}. Our result answers this question and shows that it does not since
in the presence of a weak bidder the revenue of the $k$-lookahead
auction on the original market is actually identical to the revenue
of the revenue-maximizing mechanism on the smaller market. Furthermore,
our result is asymptotically tight since
the revenue of the $k$-lookahead auction is at least  $\frac{e^{1-\frac{1}{k}}}{e^{1-\frac{1}{k}}+1}$ of the revenue of the optimal auction \cite{chen2011approximation},
and this expression approaches $\frac{e}{e+1}$ as $k$ grows. Note that there is a gap between the bounds for any constant $k\geq 3$, closing it remains an interesting open question.

Our result also has some implications on the computational complexity
of approximating revenue maximizing auctions \cite{dobzinski2011optimal,papadimitriou2011optimal,Caragiannis:2016:LDA:2956681.2934309,chen2011approximation}. Specifically, we show that the analysis of the best currently known polynomial time truthful in expectation mechanism cannot be improved.

We now elaborate on this point. It is known that determining the revenue of the revenue-maximizing \emph{deterministic} auction for three bidders or more is NP hard within some constant factor   \cite{papadimitriou2011optimal,Caragiannis:2016:LDA:2956681.2934309}.
In contrast, a revenue maximizing \emph{truthful in expectation} mechanism
can be computed in polynomial time for any fixed number of bidders \cite{dobzinski2011optimal}. Thus, for
every constant $k$ the $k$-lookahead auction can be implemented in
polynomial time. Combining with the bound of \cite{chen2011approximation},
we get that there is a polynomial time truthful in expectation
mechanism that extracts in expectation an $\frac{e}{e+1}$ fraction
of the revenue of the optimal truthful in expectation mechanism. This is the best bound known for truthful in expectation mechanisms. As we discussed, it was not known whether the analysis of the approximation ratio of the $k$-lookahead auction is tight or not. Our result implies that we cannot hope to improve the bound by improving the analysis of the $k$-lookahead auction. 

\section{Preliminaries}

We consider a single item auction setting with $n$ bidders. Each bidder $i$ has a (privately
known) value $v_i$ for the item. The values are drawn from some joint distribution $\mathcal D$. A (direct) mechanism $M$ takes a profile $v=(v_1,\ldots, v_n)$ and returns
an allocation probability and a non-negative expected price for each bidder. We use $M\left(v\right)=\left(\left(x_{1}^{M\left(v\right)},p_{1}^{M\left(v\right)}\right),\dots,\left(x_{n}^{M\left(v\right)},p_{n}^{M\left(v\right)}\right)\right)$
to denote the outcome. Thus, given $v$, $M$ allocates
to bidder $i$ with probability $x_{i}^{M\left(v\right)}$ and bidder
$i$ pays $\frac{p_{i}^{M\left(v\right)}}{x_{i}^{M\left(v\right)}}$
if allocated.  A mechanism $M$ is \emph{ex-post individually
rational} (IR) if for all $v$ and $i$: $x_{i}^{M\left(v\right)}\cdot v_{i}\geq p_{i}^{M\left(v\right)}$. 

A mechanism $M$ is \emph{truthful in expectation} if for every $v_{i},v'_{i}$ and $v_{-i}$: $x_{i}^{M\left(v_{i},v_{-i}\right)}\cdot v_{i}-p_{i}^{M\left(v_{i},v_{-i}\right)}\geq x_{i}^{M\left(v'_{i},v_{-i}\right)}\cdot v_{i}-p_{i}^{M\left(v'_{i},v_{-i}\right)}$.
Notice that truthfulness should hold also for profiles that are not in the support of $\mathcal{D}$. The expected revenue of $M$ over $\mathcal D$ is $\sum_{v\in\mathcal{D}}\Pr_{\mathcal{D}}\left[v\right]\left(\sum_{i=1}^{n}p_{i}^{M\left(v\right)}\right)$.
A truthful in expectation mechanism $M$ is \emph{optimal} if the revenue of $M$ is at least that of any other mechanism $M'$. We let $rev(\mathcal D)$ denote the supremum of the revenue\footnote{The support of the distribution that we consider are infinite, so possibly no mechanism attains this supremum.} that can be extracted by a truthful mechanism when the values are distributed according to $\mathcal D$ and $rev_{\mathcal D}(M)$ the revenue of a specific mechanism $M$. We will sometimes omit the subscript $\mathcal D$ when $\mathcal D$ is known from the context.




A mechanism $M$ is \emph{monotone} if for every $i$, $v_{i},v'_{i},v_{-i}$
s.t. $v_{i}<v'_{i}$ we have that $x_{i}^{M\left(v_{i},v_{-i}\right)}\leq x_{i}^{M\left(v'_{i},v_{-i}\right)}$. It is well known that a mechanism can be implemented truthfully if and only if it is monotone. The following proposition gives the payments:
\begin{proposition}[\cite{myerson1981optimal}]\label{prop:randomized}
A mechanism $M$ is truthful in expectation if and only if $$p_i^{M(v)}=\int_0^{v_i}z\cdot \frac d {dz} x_i^{M(z,v_{-i})}dz$$
\end{proposition}


\section{A Market with a Revenue Loss of $\frac{1}{\lowercase{e}+1}$}\label{sec:k-LA}


Let $\mathcal H_n$ be a distribution over the values of $n$ bidders. We say that bidder $i$ is a weak bidder in $\mathcal H_n$ if $v_i$ is the same in every profile $(v_1,\ldots, v_n)$ in the support of $\mathcal H_n$ and furthermore we have that $v_i< v_{i'}$ for all $i\neq i'$. $\mathcal H_n$ \emph{contains a weak bidder} if there is some bidder that is weak in $\mathcal H_n$. Without loss of generality, we will assume that the weak bidder is bidder $n$. Given a distribution $\mathcal H_n$ which contains a weak bidder, let $\mathcal H_{n-1}$ be the \emph{distribution of $\mathcal H_n$ after shrinkage}: a distribution over the values of bidders $1,...,n-1$ which is obtained by sampling from $\mathcal H_n$ and ignoring the value of the weak bidder $n$. Our main result analyzes the revenue loss due to the shrinkage:
\begin{theorem}\label{thm:cplx}
For every $n\geq 3$ and $\delta>0$, there exist a distribution $\mathcal H_n$ that contains a weak bidder and a distribution $\mathcal H_{n-1}$ of $\mathcal H_n$ after shrinkage such that  $\frac{rev\left(\mathcal H_{n-1}\right)}{rev\left(\mathcal H_n\right)}<\frac{e}{e+1}+\delta$.
\end{theorem}
As noted in the introduction, this ratio ($\approx 0.731$) is asymptotically tight \cite{chen2011approximation}. For $n=2$ the right ratio is $\frac 1 2$ \cite{ronen2001approximating} and for $n=3$ the right ratio is $\frac{ \sqrt e}{1+ \sqrt e}$ \cite{chen2011approximation}. 

The rest of this section is devoted to proving Theorem \ref{thm:cplx} and we start by giving some intuition on the proof. We will construct $\mathcal H_{n}$, a distribution over the values of $n$ bidders with the following properties:

\begin{itemize}
\item The value of bidder $1$ will be selected from a family of equal revenue distributions, each with an expected revenue of $\frac e {e-1}$. 
\item We will always have $v_2,\ldots, v_{n-1}\approx 1$ in the support of $\mathcal H_{n-1}$.  The precise values of $v_2,\ldots, v_{n-1}$ will jointly encode a parameter $h$ that will determine the specific distribution of $v_1$.
\item The value of bidder $n$ is always fixed to $1$.
\end{itemize}

We will see that to maximize revenue one has to determine the parameter $h$ by observing $v_2,\ldots, v_{n-1}$, and offer bidder $1$ to purchase the item at a price that is a function of $h$. If bidder $1$ declines to purchase the item at the requested price, we obviously want to sell the item to one of the bidders $2,\ldots,n$ at the highest price possible, which is approximately $1$ (since the values of bidders $2,\ldots, n$ is approximately $1$). However, we will not be able to sell the item to bidders $2,\ldots, n-1$. The intuitive reason is that we've used their values in order to determine $h$. However, we may sell to bidder $n$ since we have not used $v_n$ to determine the price for bidder $1$. Thus, the revenue is maximized in a mechanism that determines the take-it-or-leave-it offer to bidder $1$ by querying the values of $v_2,\ldots, v_{n-1}$, and if bidder $1$ rejects the offer the item is sold to bidder $n$ for price $1$.

The point is that bidder $n$ is a weak bidder in $\mathcal H_n$. In particular, to maximize revenue in the distribution $\mathcal H_{n-1}$ (that is obtained from $\mathcal H_n$ by removing bidder $n$), one still has to determine $h$ by querying $v_2,\ldots, v_{n-1}$ and set accordingly a take-it-or-leave-it offer to bidder $1$. However, if bidder $1$ rejects the offer we cannot sell the item at all. The gap in the revenue between the optimal mechanisms before and after the shrinkage is therefore exactly the probability that optimal mechanisms for $\mathcal H_n$ sell the item to bidder $n$. Equivalently, this is the probability that the item is not sold at all in an optimal mechanism for the distribution $\mathcal H_{n-1}$.

\subsection{\label{subsec:A-hard-distribution}The Distribution over $n$ Bidders}

We start with defining the distribution over the values of $n$ bidders $\mathcal H_n$. The next definition will be used to determine the specific values of $v_2,\ldots, v_{n-1}$ which in turn will be used to determine the value of the parameter $h$.
\begin{definition}\label{def:balanced}Let $X_{\epsilon}^{d}$ be a random variable with support over the positive integers. Let $1>\epsilon>0$. $X_{\epsilon}^{d}$ is  $\left(\epsilon,d\right)-\emph{balanced}$ if for every integer $i\geq 1$ s.t. $i\mod d=1$:
\begin{enumerate}
\item $\Pr\left[X_{\epsilon}^{d}={i}\right]=\Pr\left[X_{\epsilon}^{d}={i+1}\right]=\dots=\Pr\left[X_{\epsilon}^{d}={i+(d-1)}\right]$.
\item $\left(1-\epsilon\right)\Pr\left[X_{\epsilon}^{d}={i}\right]=\Pr\left[X_{\epsilon}^{d}={i+d}\right]$.
\end{enumerate}
\end{definition}
Note that if $X_{\epsilon}^{d}$ is $(\epsilon,d)$-balanced
then for all $\ell\geq1$ and $j\in\left\{ 1,\dots,d\right\} $, $\Pr\left[X_{\epsilon}^{d}=\ell\right]\geq\Pr\left[X_{\epsilon}^{d}=\ell+j\right]\geq\left(1-\epsilon\right)\Pr\left[X_{\epsilon}^{d}=\ell\right]$. When $\epsilon$ is known from the context we sometimes refer to $\left(\epsilon,d\right)$-balanced random variables as $d$-balanced.  Appendix \ref{prop:balancedexistance} shows the existence of $(\epsilon,d)$-balanced distributions.

Define $\mathcal H_n$, a distribution (with parameters $d$,$\epsilon$)
over the values of $n$ bidders: let $z=\frac{e}{e-1}$, $m=\left\lfloor \frac{d}{z}-1\right\rfloor $. Let $\mathcal{D}_{0},\dots,\mathcal{D}_{m-1}$ be $m$ distributions over $\mathbb{R}$ that will be defined in Subsection \ref{subsec:The-distributions}:
\begin{itemize}
\item $v_n$ is always fixed to $1-2\epsilon$.
\item For every $2\leq i< n$, let $v_{i}=1-2\epsilon+\epsilon\cdot \Sigma_{j=1}^{X_i}\frac 1 {2^j}$ where $X_i$ is an independent $(\epsilon,d)$-balanced variable.
\item Let $h(v_{-1})=\min \left(\sum_{j=2}^{n-1}X_i\mod d,m-1\right)$. $v_1$ is distributed $\mathcal{D}_{h\left(v_{-1}\right)}$.

\end{itemize}

\subsubsection{\label{subsec:The-distributions}The Distributions $\mathcal{D}_{0},\dots,\mathcal{D}_{m-1}$}

We now describe the distributions $\mathcal{D}_{0},\dots,\mathcal{D}_{m-1}$. All are equal revenue distributions. The description is technical and might require some time to digest. However, we do note that for the analysis we will mostly refer to the simple properties that are stated in Lemma \ref{lem:calcs}.

Define the probabilities $q_{0},\dots,q_{m-2}\in\mathbb{R}$ and $\overline{q_{0}},\dots\overline{q_{m-1}}\in\mathbb{R}$:
\begin{itemize}
\item For every $y\in\left\{ 0,\dots,m-2\right\} $: $q_{y}=\frac{\left(z-1\right)d}{\left(d-y\right)\left(d-y-1\right)}$.
\item For every $y\in\left\{ 0,\dots,m-1\right\} $: $\overline{q_{y}}=1-\sum_{j=0}^{y-1}q_{j}$.
\end{itemize}
Define the values $t_{0},t_{1},\dots,t_{m-1}\in\mathbb{R}$: $t_{y}=\frac{z}{\left(1-\sum_{j=0}^{y-1}q_{j}\right)}=\frac{z}{\overline{q_{y}}}$.
Note that $z=t_{0}<t_{1}<\dots<t_{m-1}$. We now define $\mathcal{D}_{0},\dots,\mathcal{D}_{m-1}$. First, $D_0(x)= t_0=z$ with probability $1$. For every $1\leq y \leq m-1$:

\[
\mathcal{D}_{y}\left(x\right)=\begin{cases}
t_0 & \text{w.p. } q_{0}\\
t_1 & \text{w.p. }q_{1} \\
\vdots & \vdots\\
t_{y-1} & \text{w.p. }q_{y-1} \\
t_y & \text{w.p. }\overline{q_{y}}
\end{cases}
\]
Notice that the support of each distribution $\mathcal{D}_{i+1}$ contains the support of $\mathcal{D}_{i}$ and one additional value. Moreover, all distributions are equal revenue distributions (see Lemma \ref{lem:calcs}). We note that:
\begin{itemize}
\item For every $y\in\left\{ 0,\dots,m-1\right\} $ and every $j\in\left\{ y,\dots,m-1\right\} $: $\overline{q_{y}}=\Pr\left[v_{1}=t_{y}|h\left(v_{-1}\right)=y\right]=\Pr\left[v_{1}\geq t_{y}|h\left(v_{-1}\right)=j\right]$.
\item For every $y\in\left\{ 0,\dots,m-2\right\} $: $q_{y}=\Pr\left[v_{1}=t_{y}|y<h\left(v_{-1}\right)\leq m-1\right]$.
\end{itemize}
\begin{proof}
We prove only the first claim, the second part is immediate from the definitions. For the first bullet point, let $y\in\{1,\ldots,m-1\}$ and let $m-1\geq j\geq l$. Consider $\mathcal{D}_{y}$ and $\mathcal{D}_{j}$. The probability that bidder $1$'s value is lower than $t_y$ if his value is distributed $\mathcal{D}_{y}$ or $\mathcal{D}_{j}$ is $\sum_{k=0}^{j-1}q_{k}$, as his probability for every value $t_k$ is $q_k$ for every $k\in\{0,\ldots,y-1\}$. Thus the probability of his valuation being higher is $1-\sum_{k=0}^{j-1}q_{k}$. For $\mathcal{D}_{y}$, the only higher value possible is $t_y$.
\end{proof}
Next we prove some simple and useful facts related to $\mathcal{D}_{0},\dots,\mathcal{D}_{m-1}$.
\begin{lemma}\label{lem:calcs}
For $d\geq4$ the following holds:
\begin{enumerate}
\item For all $y\in\left\{ 0,\dots,m-1\right\} $, $\overline{q_{y}}\cdot t_{y}=z$ (i.e., $\mathcal{D}_{0},\dots,\mathcal{D}_{m-1}$ are equal
revenue distributions). 
\item For all $y\in\left\{ 0,\dots,m-2\right\} $, $\overline{q_{y+1}}+q_{y}=\overline{q_{y}}$.
\item For all $y\in\left\{ 1,\dots,m-1\right\} ,$ $\sum_{j=0}^{y-1}q_{j}=\frac{\left(z-1\right)}{\left(d-y\right)}y$.
\item For all $y\in\left\{ 0,\dots,m-1\right\} $, $\overline{q_{y}}=\frac{d-z\cdot y}{d-y}$
.
\item For all $y\in\left\{ 0,\dots,m-2\right\} $, $\overline{q_{y}}+\left(d-y-1\right)q_{y}=z=\overline{q_{y}}\cdot t_{y}$.
\end{enumerate}
\begin{proof}\ 
\begin{enumerate}
\item $\overline{q_{y}}\cdot t_{y}=\overline{q_{y}}\cdot \frac{z}{\overline{q_{y}}}=z$.
\item $\overline{q_{y+1}}+q_{y}=\left(1-\sum_{j=0}^{y}q_{y}+q_{y}\right)=\left(1-\sum_{j=0}^{y-1}q_{y}\right)=\overline{q_{y}}$.
\item $\sum_{j=0}^{y-1}q_{j}=\sum_{j=0}^{y-1}\frac{\left(z-1\right)d}{\left(d-j\right)\left(d-j-1\right)}=\left(z-1\right)d\cdot\sum_{j=0}^{y-1}\frac{1}{\left(d-j\right)\left(d-j-1\right)}=\left(z-1\right)d\cdot\sum_{r=d-\left(y-1\right)}^{d}\frac{1}{r\left(r-1\right)}=\left(z-1\right)d\cdot\sum_{r=d-\left(y-1\right)}^{d}\left(\frac{1}{r-1}-\frac{1}{r}\right)=\left(z-1\right)\cdot d\cdot\left(\frac{1}{d-y}-\frac{1}{d}\right)=\left(z-1\right)\cdot d\cdot\frac{y}{d\left(d-y\right)}=\frac{\left(z-1\right)}{\left(d-y\right)}\cdot y$.
\item If $y=0$ then $\overline{q_{y}}=1=\frac{d-z\cdot0}{d-0}$. If $y>0$ then, using property $3$, $\overline{q_{y}}=\left(1-\sum_{j=0}^{y-1}q_{j}\right)=1-\frac{\left(z-1\right)}{d-y}\cdot y=\frac{d-z\cdot y}{d-y}$.
\item $\overline{q_{y}}+\left(d-y-1\right)q_{y}=\frac{d-z\cdot y}{\left(d-y\right)}+\left(d-y-1\right)\frac{\left(z-1\right)d}{\left(d-y\right)\left(d-y-1\right)}=\frac{d-z\cdot y+z\cdot d-d}{\left(d-y\right)}=\frac{z\cdot\left(d-y\right)}{\left(d-y\right)}=z$.
\end{enumerate}
\end{proof}
\end{lemma}
\begin{claim}
$\mathcal{D}_{0},\dots,\mathcal{D}_{m-1}$ are valid distributions.
\end{claim}
\begin{proof}
$\mathcal{D}_{0}$ is valid, as it is a probability distribution
over one value with probability $1$. For $\mathcal{D}_{j}$ where
$1\leq j\leq m-1$, we need to show that $1>q_{j}>0$, $1>\overline{q_{j}}>0$,
$\sum_{i=0}^{j-1}q_{i}<1$ and that $\sum_{i=0}^{j-1}q_{i}+\overline{q_{j}}=1$.
By definition $\sum_{i=0}^{j-1}q_{i}+\overline{q_{j}}=\sum_{i=0}^{j-1}q_{i}+\left(1-\sum_{i=0}^{j-1}q_{i}\right)=1$.
We will show that $\sum_{i=0}^{j-1}q_{i}<1$ and the rest follows,
as $\overline{q_{i}}=1-\sum_{j=0}^{i-1}q_{j}$ and $q_{i}$ are positive. Using property $3$:
\[
\sum_{i=0}^{j-1}q_{i}\leq \frac{\left(z-1\right)}{\left(d-(m-1)\right)}(m-1)\overset{m=\left\lfloor \frac{d}{z}\right\rfloor -1}{=}\left(z-1\right)\frac{\left\lfloor \frac{d}{z}-1\right\rfloor }{d-\left\lfloor \frac{d}{z}\right\rfloor +2}<\left(z-1\right)\frac{\frac{d}{z}}{d-\left\lfloor \frac{d}{z}\right\rfloor }=\frac{d-\frac{d}{z}}{d-\left\lfloor \frac{d}{z}\right\rfloor }\leq1
\]
\end{proof}

\subsection{Outline of the Proof of Theorem \ref{thm:cplx}}

We will now derive the main result by applying a few lemmas. The proof of Lemma \ref{lem:firstLemComplex} is in Subsection \ref{subsec:Prooflem1} and the proof of Lemma \ref{lem:SecondLemComplex} can be found in Subsection \ref{subsec:Prooflem2}. 

We first want to claim that to maximize revenue we should not sell the item to bidders $2,\ldots, n-1$. We will do so in two steps. In Lemma \ref{lem:firstLemComplex} we claim that we can focus on \hp \space mechanisms (a technical notion that we will shortly define). Claim \ref{lem:SecondLemComplex} shows that \hp \space mechanisms maximize the revenue by not selling to bidders $2,\ldots, n-1$.

\begin{definition}
A mechanism for $n-1$ bidders is \emph{\hp} if in every profile $v\in \mathcal H_{n-1}$ where bidder $1$ is allocated with positive probability it holds that $v_{1}\geq t_{h(v_{-1})}$.
\end{definition}

The next lemma shows that we can alter any mechanism to be \hp \space without changing the revenue (with respect to $\mathcal H_{n-1}$). The heart of this transformation is based on the observation that by the construction of $\mathcal H_{n-1}$ given the values of bidders $2,\ldots, n-1$, bidder $1$'s value is drawn from an equal revenue distribution. We use this observation as follows (for simplicity we assume in this paragraph that the mechanism is deterministic): if, given $v_2,\ldots, v_{n-1}$, bidder $1$ is offered to buy the item at price $p<t_{h(v_{-1})}$, set the price in these instances to $p=t_{h(v_{-1})}$. Note that the expected revenue remains the same because of the equal revenue distribution and since the allocation function to bidders $2,\ldots,n-1$ is unchanged. Feasibility is not violated by the transformation, since in every instance in which bidder $1$ purchases the item after the transformation he also purchases the item in the original mechanism, before the transformation. Finally, notice that given $v_{-1}$ the possible values of bidder $1$ in the support of $\mathcal H_{n-1}$ are at most $t_{h(v_{-1})}$, so the item will never be sold if $p>t_{h(v_{-1})}$.

\begin{lemma} \label{lem:firstLemComplex}
Let $M_{n-1}$ be a mechanism for $n-1$ bidders. Then, there exists a \hp \space mechanism $M_{n-1}'$ such that $rev_{\mathcal H_{n-1}}(M_{n-1}')\geq rev_{\mathcal H_{n-1}}(M_{n-1})$.
\end{lemma}
\begin{lemma}\label{lem:SecondLemComplex}
Let $M_{n-1}$ be a \hp \space mechanism for $n-1$ bidders. Then, there exists a \hp \space mechanism $M'_{n-1}$ such that for every $i=2, \ldots, n-1$ and $v\in \mathcal H_{n-1}$ we have that $x_{i}^{M_{n-1}(v)}=0$ (i.e., bidders $2,\ldots, n-1$ are never allocated) and with $rev_{\mathcal H_{n-1}}(M_{n-1}')\geq rev_{\mathcal H_{n-1}}(M_{n-1})$. 
\end{lemma}
The two lemmas give us that the revenue of every mechanism for $\mathcal H_{n-1}$ is dominated by a mechanism that only allocates to bidder $1$. This allows us to bound the revenue that can be obtained by any mechanism:
\begin{corollary}\label{cor:defkla}
Let $M_{n-1}$ be a mechanism for $n-1$ bidders. Then, $rev_{\mathcal H_{n-1}}(M_{n-1})\leq \frac{e}{e-1}$.
\end{corollary}
\begin{proof}
By Lemmas \ref{lem:firstLemComplex} and \ref{lem:SecondLemComplex}, there exists a mechanism $M'_{n-1}$ such that the revenue of $M_{n-1}$ is at most the revenue of $M'_{n-1}$, and $M'_{n-1}$ only allocates to bidder $1$. Now, for every realization of $v_{-1}$, bidder $1$'s value is distributed according to an equal revenue distribution with revenue $\frac{e}{e-1}$. Thus, $rev_{\mathcal H_{n-1}}\left(M_{n-1}\right)\leq rev_{\mathcal H_{n-1}}\left(M'_{n-1}\right)\leq \frac{e}{e-1}$.
\end{proof}
The next lemma shows a simple connection between the revenue of mechanisms for $\mathcal H_{n-1}$ and mechanisms for $\mathcal H_n$:

\begin{lemma}\label{lem:epsvalidcplx}
Let $M_{n-1}$ be a mechanism for $n-1$ bidders. There exists a mechanism $M_n$ for $n$ bidders such that: $rev_{\mathcal H_{n}}\left(M_n\right)\geq rev_{\mathcal H_{n-1}}(M_{n-1})+\Pr_{v\sim\mathcal H_n}\left[v_1\neq t_{h\left(v_{-1}\right)}\right]\cdot\left(1-2\epsilon\right)$.
\end{lemma}
\begin{proof}
Given $M_{n-1}$, let $M'_{n-1}$ be the \hp \space mechanism that allocates the item only to bidder $1$ and has at least the same revenue of $M_{n-1}$, as guaranteed by Lemmas \ref{lem:firstLemComplex} and \ref{lem:SecondLemComplex}.

Consider the following mechanism $M_n$: if $v_{1}\geq t_{h(v_{-1})}$ then allocate to bidder $1$ with probability $x_1^{M'_{n-1}}(v)$ and charge $x_1^{M'_{n-1}}(v)\cdot t_{h(v_{-1})}$. Otherwise, allocate to bidder $n$ and charge $1-2\epsilon$ (if $v_n\geq 1-2\epsilon$, else do not allocate the item at all). The mechanism is clearly truthful and ex-post IR. The revenue is $rev_{\mathcal H_{n-1}}\left(M'_{n-1}\right)+\Pr_{v\sim\mathcal H_n}\left[v_1\neq t_{h\left(v_{-1}\right)}\right]\cdot\left(1-2\epsilon\right)\geq rev_{\mathcal H_{n-1}}\left(M_{n-1}\right)+\Pr_{v\sim\mathcal H_n}\left[v_1\neq t_{h\left(v_{-1}\right)}\right]\cdot\left(1-2\epsilon\right)$.
\end{proof}
\begin{corollary}\label{cor:fin}$
\frac{rev\left(\mathcal H_{n-1}\right)}{rev\left(\mathcal H_{n}\right)}\leq\frac{\frac{e}{e-1}}{\frac{e}{e-1}+\frac{1}{e-1}\left(\sum_{j=d-\left\lfloor \frac{d}{z}\right\rfloor +1}^{d}\frac{1}{j}\right)\left(1-2\epsilon\right)}$.
\end{corollary}
\begin{proof}
By Corollary \ref{cor:defkla}, Lemma \ref{lem:epsvalidcplx} and a simple calculation of $\Pr_{v\sim\mathcal H_n}\left[v_1\neq t_{h\left(v_{-1}\right)}\right]$:
\begin{align}
\Pr_{v\sim\mathcal H_n}\left[v_1\neq t_{h\left(v_{-1}\right)}\right] =& \sum_{j=0}^{m-2}\Pr\left[v_{1}=t_{j}\wedge h\left(v_{-1}\right)\neq j\right]=\sum_{j=0}^{m-2}\Pr\left[v_{1}=t_{j}|h\left(v_{-1}\right)>j\right]\cdot\Pr\left[h\left(v_{-1}\right)>j\right]\nonumber \\
= & \sum_{j=0}^{m-2}\Pr\left[v_{1}=t_{j}|h\left(v_{-1}\right)>j\right]\cdot\left(1-\Pr\left[h\left(v_{-1}\right)\leq j\right]\right)\nonumber \\
= & \sum_{j=0}^{m-2}\overset{=q_{j}=\frac{d\cdot\left(z-1\right)}{\left(d-j-1\right)\left(d-j\right)}}{\overbrace{\Pr\left[v_{1}=t_{j}|h\left(v_{-1}\right)>j\right]}}\cdot  \overset{\frac {d-j-1} d}{\overbrace{\left(1-\frac{j+1}{d}\right)}}\nonumber \\
= & \left(z-1\right)\sum_{j=0}^{m-2}\frac{1}{\left(d-j\right)}=\overset{m=\left\lfloor \frac{d}{z}-1\right\rfloor }{=}\left(z-1\right)\sum_{j=d-\left\lfloor \frac{d}{z}\right\rfloor +1}^{d}\frac{1}{j}\nonumber
\end{align}
\end{proof}
We can now finish the proof of Theorem \ref{thm:cplx}. Applying Lemma \ref{lem:limexp} (in the appendix), we have that $\lim_{d\rightarrow\infty,\epsilon\rightarrow0}\frac{rev\left(\mathcal H_{n-1}\right)}{rev\left(\mathcal H_{n}\right)}\leq\frac{\frac{e}{e-1}}{\frac{e}{e-1}+\left(\frac{e}{e-1}-1\right)}=\frac{e}{2e-\left(e-1\right)}=\frac{e}{e+1}$.
Thus, by Corollary \ref{cor:fin}, for every $\delta>0$ there exist $d\geq4$
and $\epsilon>0$ such that $\frac{rev\left(\mathcal H_{n-1}\right)}{rev\left(\mathcal H_{n}\right)}<\frac{e}{e+1}+\delta$.

\subsection{\label{subsec:Prooflem1}Proof of Lemma \ref{lem:firstLemComplex}}

Consider a mechanism $M_{n-1}$, and let $M'_{n-1}$ be the following \hp \space mechanism for $n-1$ bidders. The allocation and payments of bidders $2,\ldots , n-1$ remain the same as in $M_{n-1}$. The allocation and payment of bidder $1$ are defined as follows:
\[
\begin{cases}
\left(x_{1}^{M_{n-1}\left(t_{h(v_{-1})},v_{-1}\right)},t_{h(v_{-1})}\cdot x_{1}^{M_{n-1}\left(t_{h(v_{-1})},v_{-1}\right)}\right) & v_{1}\geq t_{h(v_{-1})};\\
\left(0,0\right) & \text{otherwise.}
\end{cases}
\]
\begin{claim}
$M'_{n-1}$ always outputs a feasible allocation and is truthful in expectation.
\end{claim}
\begin{proof}
Notice that for every $v$ the allocation is valid since bidders $2,\ldots, n-1$ are allocated identically and the allocation probability of bidder $1$ does not increase: when $v_{1}< t_{h(v_{-1})}$ then bidder $1$ is not allocated at all in $M'_{n-1}$. When $v_{1}= t_{h(v_{-1})}$ then the allocation is identical. When $v_{1}> t_{h(v_{-1})}$ the allocation probability of bidder $1$ in $M_{n-1}$ is by monotonicity at least $x_{1}^{M_{n-1}}\left(t_{h(v_{-1})},v_{-1}\right)$. This last expression is the allocation probability of bidder $1$ in $M'_{n-1}$.

As for truthfulness, clearly, for bidders $2,\dots,n-1$ the mechanism is truthful in expectation, as we have not changed their allocation probabilities or payments. As for bidder $1$, his allocation function is clearly monotone and the payments are according to Proposition \ref{prop:randomized}. Hence $M'_{n-1}$ is truthful in expectation.
\end{proof}

It is left is to show that $rev\left(M'_{n-1}\right)\geq rev\left(M_{n-1}\right)$. The payments of bidders $2, \ldots, n-1$ are identical in both mechanisms, thus it remains to show that the expected payment of bidder $1$ has not decreased. We show that for every fixed $v_{-1}$. Let $h(v_{-1})=y$ and $x_1=x_1^{M_{n-1}(t_y,v_{-1})}$. When the values of all other bidders are $v_{-1}$, bidder $1$'s expected payment is $\Pr[v_1=t_y|v_{-1}]\cdot t_y\cdot x_1=\overline{q_y}\cdot t_y\cdot x_1=z\cdot x_1$.

Since the allocation function of $M_{n-1}$ is monotone, for every $v'_1\leq v_1$, $M_{n-1}$ allocates to bidder $1$ with probability at most $x_1$. We now bound from above the revenue of the optimal revenue of a mechanism $M''$ for a single bidder that is distributed $\mathcal D_y$ and never allocates with probability larger than $x_1$. We will show that the revenue of $M''$ is at most $z\cdot x_1$, which will imply that $rev(M'_{n-1})\geq rev(M_{n-1})$. To see this, recall that by \cite{mehta2004randomized,dobzinski2011optimal}, any truthful in expectation mechanism for a single bidder can be implemented as a universally truthful mechanism: a distribution over deterministic mechanisms. Now, for a single bidder a deterministic mechanism is simply a take-it-or-leave-it price. Since $\mathcal D_y$ is an equal revenue distribution with revenue $z$, and since the maximum allocation is $x_1$, the revenue of $M_{n-1}$ is at most $z\cdot x_1$, as needed.

\subsection{\label{subsec:Prooflem2}Proof of Lemma \ref{lem:SecondLemComplex}}

The proof plan is to take the \hp \space mechanism $M_{n-1}$ and look for profiles in the support of $\mathcal H_{n-1}$ in which some bidder $i>1$ is allocated with positive probability. We will obtain $M'_{n-1}$ by ``fixing'' those profiles, essentially by shifting the allocation probability of bidder $i$ to bidder $1$ and then showing that the shifting has not decreased the revenue. 

We will transform $M_{n-1}$ to $M'_{n-1}$ in steps: we will find a ``minimal'' profile $v'$ (in a sense that will be formally defined) in which the allocation probability of some bidder $i>1$ is positive. We will ``fix'' this profile as well as some ``neighboring'' profiles and continue. Notice that in principle there might be an infinite number of profiles that require a fix. Fortunately, the support of $\mathcal H_{n-1}$ is countable so the process is guaranteed to eventually reach every specific profile.

Now for the formal description. We first apply a ``pre processing step'': consider some profile $v$ which is not in the support of $\mathcal H_{n-1}$. Suppose that there is some bidder $i$ with a non-zero allocation probability. We set the allocation probability of this bidder $i$ to $0$ if for each $v'_i<v_i$ it holds that $v'=(v'_i,v_{-i})$ is not in $\mathcal H_{n-1}$ or the allocation probability of bidder $i$ is $0$. 
Note that this preprocessing does not affect monotonicity and that the revenue does not decrease.

We say that a profile $v$ in the support of $\mathcal H_{n-1}$ is \emph{problematic} if for some $i>1$, $x_i^{M_{n-1}(v)}>0$ and $v_1=t_{h(v_{-1})}$. 

We first claim that if there is some $v$ in the support of $\mathcal H_{n-1}$ and some bidder $i>1$ with $x_i^{M_{n-1}(v)}>0$, then there exists a problematic profile. To see this, let $v^1$ be the minimal profile $(v_1,\ldots, v_{i-1},s,v_{i+1},\ldots, v_{n})$ which is in the support of $\mathcal H_{n-1}$ that is strictly bigger coordinate-wise than $v$, let $v^2$ be the minimal profile $(v_1,\ldots, v_{i-1},s',v_{i+1},\ldots, v_{n})\in \mathcal H_{n-1}$ which is strictly bigger than $v^1$ and so on (notice that the instances $v,v^1,v^2,\ldots$ differ only in bidder $i$'s value). Since bidder $i$'s allocation probability is positive in $v'$, by monotonicity the allocation probability of bidder $i$ is also positive in $v^k$, for every $k\geq 1$. Now observe that for some profile $j$, $1\leq j\leq d$ it holds that $v^j_1=t_{h(v^j_{-1})}$ thus $v^j$ is problematic.

If there are no problematic profiles we are already done. If there are several such profiles, let $v'$ be a problematic profile such that there is no problematic profile $v''\neq v'$ with $v''_{i'}\leq v'_{i'}$ for every $i'>1$. Note that the existence of one problematic profile $v'$ implies the existence of a ``minimal'' problematic profile since the number of profiles in the support of $\mathcal H_{n-1}$ that are dominated by $v'$ coordinate-wise is finite. Let $I$ be the set of bidders (excluding bidder $1$) with a positive allocation probability in $v'$ and let $y=h({v'_{-1}})$.

We now ``fix'' the mechanism $M_{n-1}$ by defining a \hp \space mechanism $M'$ with higher revenue. The behavior of $M'$ and $M_{n-1}$ will differ only in the profile $v'$ and ``neighboring'' profiles. We then continue fixing the other problematic profiles: we take $M'$ and fix the next problematic profile by obtaining $M''$, then fix the next problematic profile in $M''$, and so on.

Let $M'$ be the mechanism with the following allocation function: the allocations of all bidders that are not in $I\cup \{1\}$ remain the same. We set the allocation probability of every bidder $i\in I$ in $v'$ to $0$. By monotonicity, we have to set the allocation probability of bidder $i\in I$ in each of the profiles in the set $G_i=\{v|v_{-i}=v'_{-i}\text{ and }v_i\leq v'_i\}$ to $0$. 

In addition, we shift the allocation probability mass in $v'$ from each bidder $i\in I$ to bidder $1$: $x_{1}^{M'\left(v'\right)}= x_{1}^{M_{n-1}\left(v'\right)}+\Sigma_{i\in I}x_{i}^{M_{n-1}\left(v'\right)}$. To guarantee monotonicity, we set the allocation probability of bidder $1$ to $x_{1}^{M'\left(v'\right)}$, and the probability of bidders $2,\ldots,n-1$ to $0$ in every profile in $G_1=\{v|v_{-1}=v'_{-1}\text{ and }v_1> v'_1\}$. 
Notice that profiles in $G_1$ are not in the support of $\mathcal H_{n-1}$ since $v'_1=t_{h(v'_{-1})}$. The payment of bidder $1$ in the profiles $\{v'\}\cup G_1$ is $t_{y}\cdot \left (x_{1}^{M_{n-1}\left(v'\right)}+\Sigma_{i\in I}x_{i}^{M_{n-1}\left(v'\right)} \right)$. The allocation in profiles that are not in $G_1\cup \left (\cup_{i\in I} G_i\right)$ is identical to the allocation of the mechanism $M_{n-1}$.


\begin{claim}
For every bidder $i$, the allocation of $M'$ is feasible and monotone. In addition, M' is \hp.
\end{claim}
\begin{proof}
We will first show that in every profile the total allocation probability of $M'$ is at most that of $M_{n-1}$. The allocation in every valuation profile that is not in $G_1\cup\left (\cup_{i\in I} G_i\right)$ is identical to that of $M_{n-1}$, hence feasible. In $v'$ and in profiles in $G_1$, the total allocation probability in $M'$ is the total allocation probability of $M_{n-1}$ in the profile $v'$, which by feasibility of $M_{n-1}$ is at most $1$. 

As for monotonicity, given a minimal problematic vector $v'$ as above, $M'$ is monotone for all bidders except bidder $1$ and bidders in $I$, as their allocations remain the same. Bidder $1$'s allocation function is monotone since we have just changed the profile $v'$ and profiles in $G_1$: monotonicity requires that the allocation probability of bidder $1$ in each profile $v\in G_1$ will be at least his allocation probability in $v'$, and in $M'$ these two probabilities are equal. Similarly, the allocation of bidder $i\in I$ has changed only for $v'$ and profiles in $G_i$: bidder $i$'s allocation probability in $v'$ is $0$ and this probability remains $0$ when $v_i$ decreases in any other profile $v\in G_i$.

To finish the monotonicity proof, suppose that for some bidder $i$ and $v \in G_1$ which was altered due to a minimal problematic vector $v'$ the mechanism is not monotone (observe that $v_{-1}=v'_{-1}$). That is, there exists a profile $v''$ where $v_{-i}=v'_{-i}$ but $v_i >v''_i$, and bidder $i$'s allocation in $v''_i$ is greater than bidder $i$'s allocation in $v'_i$. If $v''$ is in the support of $\mathcal{H}_{n-1}$ then we reach a contradiction since $v'$ is not a minimal vector ($v''$ is ``smaller'' -- notice that minimality is only defined with respect to bidders $2,\ldots,n-1$). If $v''$ is not in the support, then by our ``preprocessing'' there must be some $u=(u_i,v''_{-i})\in \mathcal H_{n-1}$ with $u_i<v''_i$ in which the allocation probability of bidder $i$ is not $0$. But similarly to before this $u$ is in contradiction to the minimality of $v'$.

Finally, in every allocation $v\in G_i$ ($i>1$) the allocation is identical to that of $M_{n-1}$, except that the allocation probability of bidder $i$ is $0$. It is also not hard to see that if $M_{n-1}$ is \hp \space then $M'$ is \hp \space as well.
\end{proof}

The claim guarantees that the allocation function is monotone, hence there are payments that make it truthful in expectation. The next claim analyzes these payments:
\begin{claim}\label{claim-payments}
For every profile $v\neq v', v\in \mathcal H_{n-1}$: $p^{M_{n-1}(v)}_1=p^{M'(v)}_1$. For the profile $v'$,  $p^{M'(v)}_1=p^{M_{n-1}(v')}_1+t_y\cdot \Sigma_{i\in I}x_i^{M_{n-1}(v')}$. For every bidder $i>1$ and profile $v\notin \{v'\}\cup G_i$, $p^{M_{n-1}(v)}_i\leq p^{M'(v)}_i$.
\end{claim}
\begin{proof}
The claim regarding the payments of bidder $1$ is a direct consequence of Proposition \ref{prop:randomized}. As for the payments of bidder $i>1$, recall that by Proposition \ref{prop:randomized}, for any mechanism $M$ it holds that $p_i^{M(v)}=\int_0^{v_i}z\cdot \frac d {dz} x_i^{M(z,v_{-i})}dz$. That is, when drawing the allocation probability as a function of the value, the payment is the area between the allocation curve and the $y$-axis. Reducing the allocation probability at some points clearly increases the payments.
\end{proof}

We show that $rev\left(M'\right)>rev\left(M_{n-1}\right)$ by proving that even though $M'$ extracts less revenue from bidders in $I$ in $v'$
and profiles in $\cup_{i\in I}G_i$ it compensates by extracting more revenue from bidder $1$ in $v'$. 

Towards this end, consider some bidder $i\in I$ and denote $s^0=v'_i$. Fixing $v'_{-i}$, let $s^{-1}$ denote the highest value that bidder $i$ gets in $v\in \mathcal H_{n-1}$ subject to the constraints that $v_{-i}=v'_{-i}$ and $v_i<v'_i$. Define $s^{-j}$ similarly but with respect to the $j$'th highest value that bidder $i$ gets. Note that $h(v_2,\ldots, v_{i-1},s^{-d}, v_{i+1},\ldots, v_{n-1})=y$ and that in the valuation profile $(s^{-d},v'_{-i})$ bidder $i$'s allocation probability is $0$, otherwise we get a contradiction to the minimality of $v'$.

$M'$ extracts more revenue than $M_{n-1}$ from bidder $1$ in the valuation profile $(s^0,v'_{-i})$ due to the shifted mass from bidder $i$ to bidder $1$. By our discussion above, $M'$ potentially extracts less revenue than $M_{n-1}$ precisely in the valuation profiles $(s^{-j}, v'_{-i})$, for $j=1,\ldots, d-1$ where the allocation probability of bidder $i$ has possibly decreased. We conservatively assume the in profiles that are not in $\{v'\}\cup G_i$ the payments of bidder $i$ in $M'$ and $M_{n-1}$ are the same (by Claim \ref{claim-payments} they might be higher in $M'$). In fact, it is enough to focus on instances $(s^{-j}, v'_{-i})$ for $j \geq d-y+1$ since otherwise the probability (given the values of the other bidders) that $v_1=t_y$ is $0$. 

We will show that $x_{i}^{M_{n-1}(v')}\cdot \left(\Pr\left[v'\right]\left(t_{y}-v'_{i}\right)\right)>x_{i}^{M_{n-1}(v')}\cdot \sum_{j=d-y+1}^{d-1}\Pr\left[\left(s^{-j},v'_{-i}\right)\right]\cdot v'_{i}$ (i.e., the additional revenue due to the mass shifted from bidder $i$ to bidder $1$ is higher than the revenue loss due to not allocating to bidder $i$ in $\{v'\}\cup G_i$). Notice that by individual rationality we use the value of bidder $i$ as an upper bound to his payment, and by monotonicity we have that $x_{i}^{M_{n-1}(v')}$ is an upper bound to the allocation probability of $M_{n-1}$ to bidder $i$ in every profile in $G_i$. 


Recall that by the definition of $\mathcal H_{n-1}$:
\begin{align*}
\Pr\left[\left(s^0,v'_{-i}\right)\right] & =\overset{\overline{q_{y}}}{\overbrace{\Pr\left[v_{1}=t_{y}|h\left((v'_2,\ldots)\right)=y\right]}}\cdot \Pr\left[(v'_2,\ldots, v'_{i-1},s^0,v'_{i+1},\ldots, v'_{n-1})\right]\\
 & =\overline{q_{y}}\cdot\Pr\left[s^0\right]\cdot\Pi_{t\in\left\{ 2,\dots,n-1\right\} \backslash \{i\}} \Pr\left[v'_{t}\right]
\end{align*}
Where the last equality is by the fact that the values of bidders $2,\dots,n-1$
are $\left(n-2\right)$-wise independent. Additionally, for every
$j\in\left\{ d-y+1,\dots,d-1\right\} $:
\begin{align*}
\mathbb{\Pr}\left[\left(v'_{1},\dots,s^{-j},\dots,v'_{n-1}\right)\right] & =\Pr\left[v_{1}=t_{y}|h\left(\left(v'_{2},\dots,s^{-j},\dots,v'_{n-1}\right)\right)>y\right]\cdot\Pr\left[\left(v'_{2},\dots,s^{-j},\dots,v'_{n-1}\right)\right]\\
 & =q_{y}\cdot\Pr\left[s^{-j}\right]\cdot\Pi_{t\in\left\{ 2,\dots,n-1\right\} \backslash \{i\}}\Pr\left[v'_{t}\right]
\end{align*}
Thus:
\begin{align}
 \Pr&\left[v'\right]\left(t_{y}-v'_{i}\right)-\sum_{j=d-y+1}^{d-1}\Pr\left[\left(s^{-j},v'_{-i}\right)\right]\cdot v'_{i}\nonumber \\
= & \Pi_{t\in\left\{ 2,\dots,n-1\right\} \backslash \{i\}}\Pr\left[v'_{t}\right]\cdot\left(\overline{q_{y}}\cdot\overset{\geq\left(1-\epsilon\right)\Pr\left[s^{-d}\right]}{\overbrace{\Pr\left[s^0\right]}}\cdot\left(t_{y}-\overset{<1}{\overbrace{v'_{i}}}\right)-\sum_{j=d-y+1}^{d-1}\overset{<1-\epsilon}{\overbrace{v'_{i}}}\cdot q_{y}\cdot\overset{\leq\Pr\left[s^{-d}\right]}{\overbrace{\Pr\left[s^{-j}\right]}}\right)\label{eq-*2}\\
> & \Pi_{t\in\left\{ 2,\dots,n-1\right\} \backslash \{i\}}\Pr\left[v'_{t}\right]\cdot\left(1-\epsilon\right)\cdot\Pr\left[s^{-d}\right]\left(\overline{q_{y}}\cdot\left(t_{y}-1\right)-\sum_{j=1}^{d-y-1}q_{y}\right)\label{eq-*3}\\
= & \Pi_{t\in\left\{ 2,\dots,n-1\right\} \backslash \{i\}}\Pr\left[v'_{t}\right]\cdot\left(1-\epsilon\right)\left[s^{-d}\right]\left(\overset{z}{\overbrace{\overline{q_{y}}\cdot t_{y}}}-\left(\overset{z}{\overbrace{\overline{q_{y}}+\left(d-y-1\right)q_{y}}}\right)\right)\nonumber \\
= & 0 \label{eq-*4}
\end{align}
(\ref{eq-*2})
is by Definition \ref{def:balanced}, as $\Pr\left[s^{0}\right]\geq\left(1-\epsilon\right)\Pr\left[s^{-d}\right]$
and $\Pr\left[s^{-j}\right]\leq\Pr\left[s^{-d}\right]$
because the value of bidder $i$ is distributed $1-2\epsilon+\epsilon\cdot \Sigma_{j=1}^{X_i}\frac{1} {2^j}$ and ${X_i}$ is $d$-balanced. (\ref{eq-*3}) is by the
fact that $1>\left(1-\epsilon\right)>v_{i}$. (\ref{eq-*4}) is
by Lemma \ref{lem:calcs}, specifically $\overline{q_{y}}+\left(d-y-1\right)q_{y}=z$.
This shows that the revenue of $M'$ is higher than that of $M_{n-1}$.

\subsubsection*{Acknowledgments}

The second author was partially supported by BSF grant no. 2016192.

\bibliographystyle{plain}
\bibliography{k-lookahead}

\appendix

\section{Missing Proofs}

\subsection{Existence of Balanced Distributions}
\begin{proposition}
	\label{prop:balancedexistance}For every $d\geq1$,$0<\epsilon<1$
	there exists an $\left(\epsilon,d\right)$-balanced distribution.
\end{proposition}
Let $X_{\epsilon}^{d}$ be the following random variable (for $0<\epsilon<1$ and $k\geq 1$): $\Pr\left[X_{\epsilon}^{d}=k\right]=\frac{\epsilon\left(1-\epsilon\right)^{\left\lceil \frac{k}{d}\right\rceil -1}}{d}$. Notice that for $d=1$ this is a regular Geometric distribution with $p=\epsilon$. We show that $X_{\epsilon}^{d}$ is an $\left(\epsilon,d\right)$ balanced distribution (as defined in Definition \ref{def:balanced}):
\begin{enumerate}
\item Clearly, the positive integers are the support.
\item $X_{\epsilon}^{d}$ is well defined. For every positive integer $k$, $0<\Pr\left[k\right]<1$, and $\sum_{j=1}^{\infty}\frac{\epsilon\left(1-\epsilon\right)^{\left\lceil \frac{j}{d}\right\rceil -1}}{d}=1$.
\item For $i=c\cdot d+1$ (where $c\in\left(\mathbb{N}\cup\left\{ 0\right\} \right)$,
i.e. $i\,\mod d=1$): 
\begin{enumerate}
\item For $j\in\left\{ 0,\ldots,d-1\right\}$ $\Pr\left[X_{\epsilon}^{d}=i+j\right]=\frac{\epsilon\left(1-\epsilon\right)^{\left\lceil \frac{c\cdot d+1+j}{d}\right\rceil -1}}{d}=\frac{\epsilon\left(1-\epsilon\right)^{c-1+\left\lceil \frac{1+j}{d}\right\rceil }}{d}\overset{1+j\leq d}{=}\frac{\epsilon\left(1-\epsilon\right)^{c}}{d}$.
\item $\Pr\left[X_{\epsilon}^{d}=i\right]=\frac{\epsilon\left(1-\epsilon\right)^{c}}{d}$
and $\Pr\left[X_{\epsilon}^{d}={i+d}\right]=\frac{\epsilon\left(1-\epsilon\right)^{c+1}}{d}=\frac{\epsilon\left(1-\epsilon\right)^{c}}{d}\left(1-\epsilon\right)=\Pr\left[X_{\epsilon}^{d}=i\right]\left(1-\epsilon\right)$.
\end{enumerate}
\end{enumerate}

\subsection{The Limit of $\sum_{j=\protect d-\left\lfloor \frac{\protect d}{z}\right\rfloor +1}^{\protect d}\frac{1}{j}$}
\begin{lemma}
\label{lem:limexp}$\lim_{d\rightarrow\infty}z+\left(z-1\right)\left(1-2\epsilon\right)\sum_{j=d-\left\lfloor \frac{d}{z}\right\rfloor +1}^{d}\frac{1}{j}=z+\left(z-1\right)\left(1-2\epsilon\right)$
\end{lemma}
\noindent Recall that for general $a$ and $b$:

\noindent 
\[
\ln\left(\frac{a+1}{b}\right)=\int_{b}^{a+1}\frac{1}{j}dj\leq\sum_{b}^{a}\frac{1}{j}\leq\int_{j=b+1}^{a}=\ln\left(\frac{a}{b+1}\right)
\]
In our case: $a=d$ and $b=d-\left\lfloor \frac{d}{z}\right\rfloor +1$.
On one hand, as $\left(\frac{a+1}{b}\right)=\frac{d+1}{d-\left\lfloor \frac{d}{z}\right\rfloor +1}$:
\[
\lim\frac{d+1}{d-\left\lfloor \frac{d}{z}\right\rfloor }=\lim\frac{d+1}{d-\frac{d}{z}+1}=\lim\frac{zd+z}{\left(z-1\right)d+z}=\frac{z}{z-1}
\]
On the other hand, as $\frac{a}{b+1}=\frac{d+1}{d-\left\lfloor \frac{d}{z}\right\rfloor +2}$:
\[
\lim\frac{d+1}{d-\left\lfloor \frac{d}{z}\right\rfloor +2}=\lim\frac{zd+z}{d\left(z-1\right)+2z}=\frac{z}{z-1}
\]
Therefore, by the Squeeze Theorem:
\[
\lim_{d\rightarrow\infty}z+\left(z-1\right)\left(1-2\epsilon\right)\sum_{j=d-\left\lfloor \frac{d}{z}\right\rfloor +1}^{d}\frac{1}{j}=z+\left(z-1\right)\left(1-2\epsilon\right)\log\left(\frac{z}{z-1}\right)=z+\left(z-1\right)\left(1-2\epsilon\right)
\]
where the last equality follows as $\log\left(\frac{z}{z-1}\right)=\log\left(\frac{\frac{e}{e-1}}{\frac{e}{e-1}-1}\right)=1$.
\end{document}